\documentclass[11pt,a4paper]{article}

\usepackage[OT1]{fontenc}
\usepackage{lmodern}
\usepackage[english]{babel}
\usepackage[utf8x]{inputenc}
\usepackage[sc]{mathpazo}
\usepackage{amsmath,amssymb,amsfonts,mathrsfs,amstext}
\usepackage[thmmarks]{ntheorem}
\usepackage{wrapfig}

\usepackage{varioref}

\usepackage{datetime}

\usepackage{mathtools}

\usepackage{array}

\usepackage{listings}
\usepackage{algpseudocode}
\usepackage{algorithm}

\usepackage{color}
\usepackage{caption}
\usepackage{subcaption}
\usepackage[noabbrev]{cleveref}
\usepackage{csquotes}
\usepackage{bm}
\usepackage{mathtools}
\usepackage{nicefrac}
\usepackage{booktabs}

\numberwithin{equation}{section}

\newtheorem{theorem}{Theorem}[section]
\newtheorem{example}[theorem]{Example}
\newtheorem{remark}[theorem]{Remark}

\newtheorem{lemma}[theorem]{Lemma}
\newtheorem{proposition}[theorem]{Proposition}

\theoremstyle{nonumberplain}
\theorembodyfont{\normalfont}
\theoremsymbol{\ensuremath{\blacksquare}}
\newtheorem{proof}{Proof}

\renewcommand{\epsilon}{\ensuremath\varepsilon}

\renewcommand{\phi}{\ensuremath{\varphi}}

\DeclareMathAlphabet{\mathpzc}{OT1}{pzc}{m}{it}

\newcommand*\Laplace{\mathop{}\!\mathbin\bigtriangleup}

\newcommand{\NORM}[1]{\left\lVert#1\right\rVert} 

\newcommand{\DEF}{\coloneqq}

\newcommand{\RR}{\mathbb{R}}
\newcommand{\CC}{\mathbb{C}}
\newcommand{\NN}{\mathbb{N}}

\newcommand{\Om}{\Omega}

\newcommand{\del}{\partial}

\newcommand{\MID}{\!\! \mid\!}

\newcommand{\eps}{\epsilon}

\newcommand{\OO}{\mathcal{O}}

\newcommand{\intd}{\mathrm{d}}

\newcommand{\inteps}{\int\limits_{-\eps}^\eps}

\newcommand{\TransT}{\mathrm{T}}

\newcommand{\Leu}{\mathrm{L}}

\title{{Optimization of Steklov-Neumann eigenvalues}}
\date{}
\author{ Habib Ammari\thanks{\footnotesize Department of Mathematics, ETH Z\"urich, R\"amistrasse 101, CH-8092 Z\"urich, Switzerland (habib.ammari@math.ethz.ch, kthim.imeri@sam.math.ethz.ch).} 
\and Kthim Imeri\footnotemark[1] 
\and Nilima Nigam \thanks{\footnotesize Department of Mathematics, Simon Fraser University, 8888 University Dr, Burnaby, BC V5A 1S6, Canada (nigam@math.sfu.ca).} }

\begin{document}
	\maketitle

\begin{abstract}
This paper examines the Laplace equation with mixed boundary conditions, the Neumann and Steklov boundary conditions. This models a container with holes in it, like a pond filled with water but partly covered by immovable pieces on the surface. The main objective is to determine the right extent of the covering pieces, so that any shock inside the container yields a resonance. To this end, an algorithm is developed which uses asymptotic formulas concerning perturbations of the partitioning of the boundary pieces. Proofs for these formulas are established. Furthermore, this paper displays some results concerning bounds and examples with regards to the governing problem.
\end{abstract}

\def\keywords2{\vspace{.5em}{\textbf{  Mathematics Subject Classification
(MSC2000).}~\,\relax}}
\def\endkeywords2{\par}
\keywords2{{35R30, 35C20.}}

\def\keywords{\vspace{.5em}{\textbf{ Keywords.}~\,\relax}}
\def\endkeywords{\par}
\keywords{{Steklov eigenvalue problem, boundary integral operators, mixed boundary conditions, perturbations formulas.}}


\newcommand{\xS}{{x_{\textrm{S}}}}
\newcommand{\GSte}{{\Gamma_{\textrm{S}}}}
\newcommand{\lam}{\lambda}
\newcommand{\SxS}{\mathrm{S}_{\xS}^\lam}
\newcommand{\lgste}{\lambda^{\GSte}}
\newcommand{\GDir}{{\Gamma_{\textrm{D}}}}
\newcommand{\GNeu}{{\Gamma_{\textrm{N}}}}
\newcommand{\GDel}{{\Gamma_{\Delta}}}

\section{Introduction}\label{Ch:Introduction}
In this work we are concerned with resonances in a simple model of the  {\it sloshing problem}, where small vertical fluid displacements are modeled by harmonic functions in a domain satisfying mixed Steklov and Neumann boundary data. As an example, we may ask at what frequency a partially-covered cup of coffee  may spill.  To be more precise, under which harmonic movement frequency of a liquid-filled container do we achieve a high upsurge \cite{IgNobelPrice}? Or let us consider a pond with water and with immovable pieces of ice on top. Given a shock inside the pond, we observe that for the right volume of the pond and the right extent of the ice-pieces we achieve a form of resonance \cite{Nazarov2015, ice-fishing}. 

\begin{wrapfigure}{r}{0.5\textwidth}
  \centering
  \includegraphics[width=0.5\textwidth]{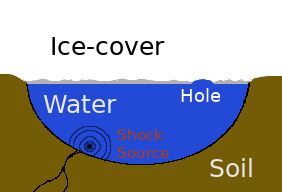}
  \caption{This is a schematic of the pond filled with water, which is encompased by soil and a cover of ice. The ice-cover has hole in it, from which water emerges after a shock inside the pond.}
  \label{fig:pond}
\end{wrapfigure}

These sort of problems are referred to as \emph{sloshing problems} and we model those problems as follows. We describe the liquid by a function $\mathrm{S}:\Om\rightarrow\RR$, where $\Om$ is a bounded domain in $\RR^2$ with boundary $\del\Om$ and $\del\Om$ expresses the container enclosing the liquid. We model a point-source forcing inside the liquid by the Dirac delta measure $\delta_\xS$ at the source location $\xS\in\Om$.  The boundary $\del\Om$ is partitioned into two parts: $\GSte$ which models a part in the container boundary which can interact with the outside media, and $\GNeu$, which models an impermeable part on the boundary. The displacement $\mathrm{S}$ satisfies the partial differential equation below, and is described in more detail in the discussion leading up to (\ref{pde:SteklovFunda}). 
 
Mathematically, the function $\mathrm{S}$ is the {\it Green's function} for the domain with the given boundary data.
As it turns out, for a countable amount of values of $\lam$ the function $\mathrm{S}$ is not defined. At these exceptional values $\mathrm{S}$ blows up. This is reminiscent of the Green's functions for mixed Dirichlet-Neumann problems \cite{HRPaperPart1, HRPaperPart2}. Our goal is to find a partition of $\del\Om$ such that a given target value is one of these exceptional values and thus will lead to a blow-up of $\mathrm{S}$. To this end, no analytic method is known and we develop a numerical approach with help of asymptotic formulas.  

Two special cases in the partitioning are discussed extensively in the literature. In the case of  Laplace's equation with pure Neumann boundary condition (when the partition where $\GNeu$ is equal to the boundary), we refer to \cite{MCMP, LPTSA}. The other case is pure Steklov problem, with  $\GSte$ being the boundary \cite{Stekloff, legacyStekloff}. The Steklov problem yields also exceptional values, so-called Steklov eigenvalues, and these eigenvalues are discussed intensively in the literature. Bounds to these are discussed in \cite{WentzellLaplace, BoundsQuasiconformalMaps, SteklovNeumannComparison, SpectralGeometry, EigenvaluesWithCurvature}, asymptotic formulas to domain perturbations can be found in \cite{ExistenceSteklovWithHole, Nazarov2015} and for boundary perturbations in \cite{Nazarov2015, AsympForBndryPerturb}. More recently, there has been active work on the asymptotics of Steklov eigenfunctions in the presence of corners, and near Steklov-Neuman junctions \cite{levitin}.

One objective is to optimize the shape of the boundary to achieve higher Steklov eigenvalues, and recent discoveries show that a rotation-symmetric, star-like domain with $k$ spikes maximizes the $k$-th Steklov eigenvalues amongst perturbations of the disc \cite{EldarsAlgo, ReformulatingSteklovProb}. We also want to mention a reformulation of the Laplace equation with Steklov boundary condition in \cite{LambertiLuigi} and existence results and bounds to the eigenvalues for the Laplace equation with mixed Steklov-Dirichlet boundary conditions in \cite{SteklovWithSmallDirichlet, DirichletWithSmallSteklov}. The Steklov-Neumann boundary condition, the main problem of interest in this paper, is also discussed in \cite{SteklovExistence, ice-fishing, SteklovNeumannComparison}.

The paper is organized as follows. In Section \ref{Ch:Prelim} we define the Steklov-Neumann problem and the associated eigenvalues. We give some examples and some bounds for the Steklov problem. These are relevant in our paper, since our algorithm in Section \ref{Ch:Algorithm} starts with the Steklov problem.

In Section \ref{Sec:EVALAsymp} we recall results about the existence of Steklov - Neumann eigenvalues and eigenfunctions. Then we prove that for small perturbation in the partitioning of the boundary $\del\Om$ the perturbed eigenvalues and eigenfunctions converge to the unperturbed eigenvalue and eigenfunctions, where we consider eigenvalues with multiplicity greater than one. This is motivated by recent work in which it has been shown that for smooth domains, the Steklov spectrum converges extremely fast to that of a disc of the same perimeter; the spectrum of the disc consists of eigenvalues of double multiplicity, except for the first eigenvalue (at zero); see \cite{SpectralGeometry} and the references therein. A slower asymptotic convergence is observed for polygonal domains. It is important, therefore, to account for the situations where higher-multiplicity Steklov eigenvalues are perturbed by inserting Neumann boundary segments.

 In Theorem \ref{thm:EVAL-Asymp} we  prove the convergence order and show explicit formulas for the highest order term in the asymptotic expression.

In Section \ref{Sec:STFct Asymp} we consider the Green's function for the Steklov - Neumann problem, that is the solution to (\ref{pde:SteklovFunda}).  Then we consider again the perturbation of the partitioning of the boundary $\del\Om$ and prove a resulting asymptotic formula for the Green's function.

In Section \ref{Ch:Algorithm} we introduce an algorithm for finding the partitioning of the boundary such that one of the Steklov-Neumann eigenvalues hits a given target value $\lam_\star\geq 0$. The actual computation of the eigenvalues is based on a boundary integral approach recently developed in \cite{nurbek}. We briefly discuss the algorithm to see where we can take advantage of the previously developed asymptotic formulas.

In Section \ref{Ch:NumImplTest} we describe how we implemented the calculation of the solution on a computer and then apply the algorithm to two domains with various numeric values to the underlying variables. 


\makeatletter
\def\moverlay{\mathpalette\mov@rlay}
\def\mov@rlay#1#2{\leavevmode\vtop{%
   \baselineskip\z@skip \lineskiplimit-\maxdimen
   \ialign{\hfil$\m@th#1##$\hfil\cr#2\crcr}}}
\newcommand{\charfusion}[3][\mathord]{
    #1{\ifx#1\mathop\vphantom{#2}\fi
        \mathpalette\mov@rlay{#2\cr#3}
      }
    \ifx#1\mathop\expandafter\displaylimits\fi}
\makeatother

\newcommand{\cupdot}{\charfusion[\mathbin]{\cup}{\cdot}}
\newcommand{\bigcupdot}{\charfusion[\mathop]{\bigcup}{\cdot}}
\newcommand{\OmBar}{\overline{\Om}}
\newcommand{\ZxS}{\mathrm{Z}_{\xS}^k}
\newcommand{\ZxSD}{\mathrm{Z}_{\mathrm{D}, \xS}^k}
\newcommand{\ZxSN}{\mathrm{Z}_{\mathrm{N}, \xS}^k}
\newcommand{\lgdir}{\lambda^{\GDir}}
\newcommand{\lgfull}{\lambda^{\del\Om}}
\newcommand{\lgempt}{\lambda^{\varnothing}}
\newcommand{\Hank}[1]{\mathrm{H}^{(1)}_{#1}}
\newcommand{\Sobo}{{H}}
\newcommand{\Sobominhalf}{{H}^{-\nicefrac{1}{2}}}
\newcommand{\Soboplushalf}{{H}^{\nicefrac{1}{2}}}
\newcommand{\Hminhalftilde}{\widetilde{H}^{-\nicefrac{1}{2}}}
\newcommand{\HminhalftildeNull}{\widetilde{H}^{-\nicefrac{1}{2}}_{\langle 0\rangle}}
\newcommand{\Hplushalfast}{{H}^{\nicefrac{1}{2}}_{\ast}}
\newcommand{\SGDk}{\mathcal{S}_\GDir^k}
\newcommand{\SGNk}{\mathcal{S}_\GNeu^k}
\newcommand{\KGNkstar}{(\mathcal{K}^{k}_\GNeu)^{\ast}}
\newcommand{\KGGkstar}{(\mathcal{K}^{k}_{\del\Om})^{\ast}}
\newcommand{\dSGDk}{\del\mathcal{S}^{k}_{\GDir}}
\newcommand{\Gk}{\Gamma^k}
\newcommand{\calAk}{\mathcal{A}^k}
\newcommand{\calAe}{\mathcal{A}_\eps}
\newcommand{\calBe}{\mathcal{B}_\eps}
\newcommand{\calAz}{\mathcal{A}_0}
\newcommand{\psiDN}{\begin{bmatrix} \psi\MID_{\GDir}\\ \psi\MID_{\GNeu}\end{bmatrix}}
\newcommand{\calAkz}{\mathcal{A}^k_0}
\newcommand{\calAke}{\mathcal{A}^k_\eps}
\newcommand{\kjzer}{k_j^0}
\newcommand{\kjeps}{k_j^\eps}

\section{Preliminaries}\label{Ch:Prelim}
Let $\Om\subset\RR^2$ be an open, simply connected, bounded domain with a smooth boundary. We define $\OmBar$ as the topological closure of $\Om$. We decompose the boundary $\del\Om\DEF\OmBar\setminus\Om$ into two parts, $\del\Om=\overline{\GSte \cupdot \GNeu}$, where $\GSte$ and $\GNeu$ are finite unions of open line segments. Then we define $(\GSte,\GNeu)$ to be a partition of $\del\Om$.

\begin{wrapfigure}{r}{0.42\textwidth}
  \centering
  \includegraphics[width=0.24\textwidth]{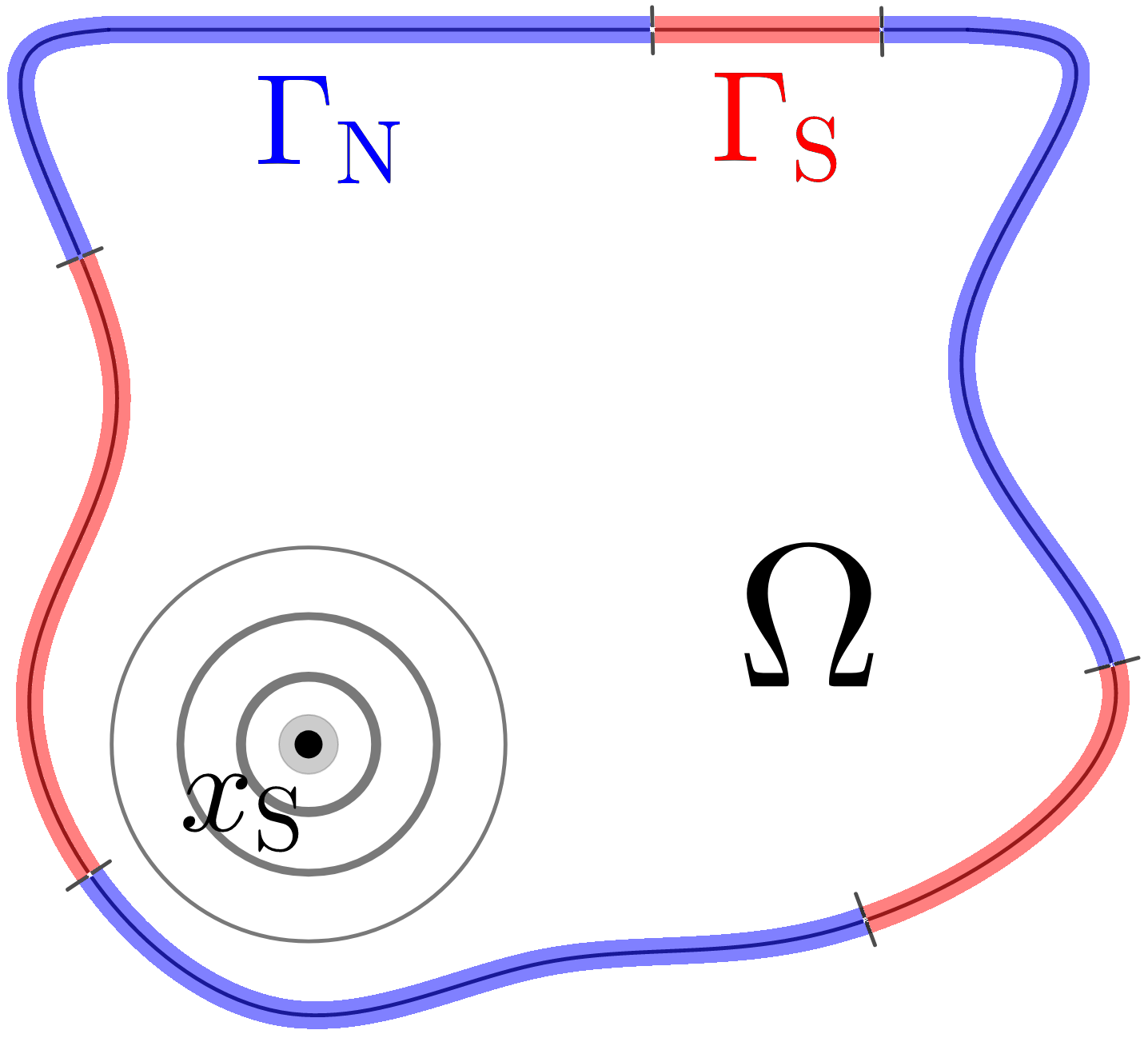}
  \caption{An $\Om$ with $\xS\in\Om$ and a partition $\del\Om=\overline{\GSte \cupdot \GNeu}$, where $\GSte$ is marked in red, $\GNeu$ in blue.}\label{fig:PreLimDomain}
  \vspace{-10pt}
\end{wrapfigure}

Let $\xS\in\Om$ and $\lam\in[0,\infty)$. Then we define the Steklov-Neumann function $\SxS: \Om\setminus\{\xS\}\rightarrow \RR$, also written as $\mathrm{S}^\lam(\xS\,, \cdot)$, as the Green's function to the mixed Steklov-Neumann problem, that is the fundamental Laplace equation with mixed, homogeneous Steklov and Neumann boundary conditions,
\begin{align}\label{pde:SteklovFunda}
	\left\{ 
	\begin{aligned}
		  \Laplace  \SxS(y)     &= \delta_0(\xS-y) \quad &&\text{for} \; &&y\in\Om\,, \\
		  \del_{\nu}\SxS(y) 	&= \lam\,\SxS(y) \quad &&\text{for} \; &&y\in\GSte \,,\\
		  \del_{\nu} \SxS(y) 	&= 0 \; &&\text{for} \; &&y\in\GNeu. \
	\end{aligned}
	\right.
\end{align}
Here $\nu_y$ denotes the outer normal at $y\in\del\Om$, $\del_{\nu_y}$ the normal derivative at $y\in\del\Om$ and $\delta_0$ the Dirac measure. In Theorem \ref{thm:decomp S} we 
prove existence of $\SxS\in\Leu^2(\Om)$ except for some exceptional values of $\lam$.

These exceptional values of $\lam$ are precisely the Steklov-Neumann eigenvalues given through
\begin{align}\label{pde:SteklovHom}
	\left\{ 
	\begin{aligned}
		 \Laplace      u		&= 0 \quad &&\text{in} \; &&\Om\,, \\
		 \del_{\nu}    u      	&= \lam \, u \quad &&\text{on} \; &&\GSte \,,\\
		 \del_{\nu}    u		&= 0 \quad &&\text{on} \; &&\GNeu \,.
	\end{aligned}
	\right.
\end{align}
Equation (\ref{pde:SteklovHom}) has a non-trivial solution $u\in \mathrm{H}^1(\Om)$ for a countable set of real values of $\lam$ \cite[Section 6 and 7]{SteklovExistence}, which we refer to as $\{\lgste_j\}_{j=1}^\infty$, so that $\lgste_1\leq\lgste_2\leq\lgste_3\leq\ldots\,$. We know that $\lgste_1 = 0$ is a simple eigenvalue and that $\lim_{j\rightarrow\infty}\lgste_j= +\infty$ for all partitions $(\GSte,\GNeu)$ of $\del\Om$.

We denote by  $\{\lgfull_j\}_{j\in\NN}$ the eigenvalues associated to the case $\GSte=\del\Omega$, i.e., the pure Steklov eigenvalues. 

For a smooth boundary, but not necessarily for non-smooth ones, and for simply connected domains, it is shown, for example in \cite{EdwardJulian}, that Steklov eigenvalues converge very fast to the spectrum of a disk of the same diameter, that is,
\begin{align*}
	\lambda_j^{\del\Omega} = \frac{2\pi}{|\del\Om|}j+\OO(j^{-\infty})\,,
\end{align*}
where the notation $\OO(j^{-\infty})$ means that the error term decays faster than any power of $j$. Moreover, $\lambda_{2j}^{\del\Omega} = \lambda_{2j+1}^{\del\Omega} +\OO(j^{-\infty})$, for $j\geq 1$. For more general smooth Riemann surfaces a similar result was established in \cite{SherPolterovich}. In the following we present some spectra for some explicit domains, starting with the unit circle.

\begin{example}\label{exp:disk}
	Let $\Om$ be the unit circle. We have that $\{\lgfull_j\}_{j=1}^\infty$ is equal to 
	$$
	\{ 0, 1, 1, 2, 2, 3, 3, 4, 4, \ldots\}\,,
	$$
	where all but the lowest element are twice degenerate and the corresponding non-normalized eigenfunctions are given in polar coordinates by \\$\{ r\,\sin(\lam\theta)\,, r\,\cos(\lam \theta)\}$. We refer to \cite{SpectralGeometry}. 
\end{example}

\begin{example}
	Let $\Om=(-1,1)\times(-1,1)$ be a square. The Steklov eigenvalues except for  0 are given in the following table:
	\begin{center}
	\begin{tabular}{ c l c }
 		$\text{Eigenspace basis}$ & $\text{Condition on $\alpha>0$}$ & $\text{Eigenvalue}$ \\ [0.5ex]
 		$\substack{\cos(\alpha\, x)\,\cosh(\alpha\, y)\\
 		\cos(\alpha\, y)\,\cosh(\alpha\, x)} $
 			& $\tan(\alpha)=-\tanh(\alpha)$ 
 			& $\alpha\,\tanh(\alpha)$ \\  
 		\hline
 		$\substack{\sin(\alpha\, x)\,\cosh(\alpha\, y)\\
 		\sin(\alpha\, y)\,\cosh(\alpha\, x)} $
 			& $\tan(\alpha)=\coth(\alpha)$ 
 			& $\alpha\,\tanh(\alpha)$ \\   
 		\hline
 		$\substack{\cos(\alpha\, x)\,\sinh(\alpha\, y)\\
 		\cos(\alpha\, y)\,\sinh(\alpha\, x)} $
 			& $\tan(\alpha)=-\coth(\alpha)$ 
 			& $\alpha\,\coth(\alpha)$ \\   
 		\hline 
 		$\substack{\sin(\alpha\, x)\,\sinh(\alpha\, y)\\
 		\sin(\alpha\, y)\,\sinh(\alpha\, x)} $
 			& $\tan(\alpha)=\tanh(\alpha)$ 
 			& $\alpha\,\coth(\alpha)$ \\ 
 		\hline 
 		$x\,y$ & $\text{none}$ & $1$ \\    
	\end{tabular}
	\end{center}
	Consider that these eigenvalues are not sorted and every eigenvalue on the very right column has multiplicity two, with their respective eigenspace basis given on the very left column, except for the eigenvalue $1$ and $0$, which are simple eigenvalues, with the eigenfunctions $f(x, y)=x\, y\,$ and $f(x, y)= 1$, respectively. Thus the sorted eigenvalues are given approximately by
	\begin{align*}
		\{
		  0,\; 0.938, 0.938,\; 1,\; 2.347, 2.347,\; 2.365, 2.365,\ldots
		\}\,.
	\end{align*}
	We refer to \cite[Section 3.1]{SpectralGeometry}.
\end{example}

\begin{example}
	Let $\Om_\eps=\{x\in\RR^2 \;|\; \eps<|x|<1\}$ be a ring with a hole of radius $\eps$. We note here that $\Om_\eps$ is not simply connected. The only radially independent eigenfunction is
	\begin{align*}
		f(r)=\bigg( \frac{-(1+\eps)}{\eps\,\log(\eps)}\bigg)\log(r)+1\,,
	\end{align*}
	with the corresponding eigenvalue
	\begin{align*}
		\lam^\eps = -\frac{1+\eps}{\eps}\log(\eps)\,.
	\end{align*}
	The other eigenvalues and eigenfunctions are given in \cite{ReformulatingSteklovProb}.
\end{example}

\begin{example}
	Let $\Om_{\eps,\, k}$ be a rotation-symmetric perturbation of the unit disk, which is the shape of a flower and which can be described through polar coordinates as $\{(r,\theta)\in\RR_{\geq 0}\times[-\pi,\pi) \;|\; r\leq 1+\eps\,\cos(k\,\theta)\}$. Let $\{\lam^0_j\}_{j=1}^\infty$ be the Steklov eigenvalues for the disk, see Example \ref{exp:disk}, and let $\{\lam^\eps_j\}_{j=1}^\infty$ be the Steklov eigenvalues for $\Om_{\eps,\, k}$. Then under some regularity assumptions \cite{AsympForBndryPerturb}, it can be shown that for $k=0$ we have
	\begin{align*}
		\lam_j^\eps = (1+\eps)^{-1}\lam_j^0\,,\quad\text{for all } j\in\NN_0\,,
	\end{align*}
	and for $k>0$ odd, or $k\neq 2\,\lam_j^0$, we have that
		$\lam_j^\eps-\lam^0_j = \OO(\eps^2)\,,$	 
	and for the remaining cases, that is for $\lam_k^\eps$ and $\lam_{k+1}^\eps$, where $\lam_k^0=\lam_{k+1}^0=j/2$, $k$ even, we have that
	\begin{align*}
		\lam_{k}^\eps-\lam_k^0 			&= \,-\eps\,\frac{j(j+1)}{4}+\OO(\eps^2) \,,\\
		\lam_{k+1}^\eps-\lam_{k+1}^0 	&= \,\eps\,\frac{j(j+1)}{4}+\OO(\eps^2) \,.
	\end{align*}
	The paper \cite{AsympForBndryPerturb} gives explicit formulas for the second order terms.
\end{example}

As for bounds to Steklov-Neumann eigenvalues, we have from \cite[Theorem 1.8]{SteklovNeumannComparison} that for all $j\in\NN$,
\begin{align*}
	\lam_j^{\GSte}\leq \frac{2\pi\,(j-1)}{|\GSte|},
\end{align*}
where $|\GSte|$ denotes the length of $\GSte$.
\cite{HPSInequ} showed that  in fact one has a strict inequality for $\lam^\GSte_3$, that is,
\begin{align*}
	\lam^\GSte_3 < 	  \frac{4\pi}{|\GSte|}\,.
\end{align*}

\section{Eigenvalue Asymptotics}\label{Sec:EVALAsymp}
In this section we derive an asymptotic formula which describes the behaviour of the spectrum when we change a small part $\GDel$ of the  Steklov boundary to  Neumann. 

To this end, we first need to prove convergence.

Let $\{ u_j^0 \}_{j\in\NN}$ be Steklov-Neumann eigenfunctions corresponding to the Steklov-Neumann eigenvalues $\{\lam_j^0\}_{j\in\NN}\DEF\{\lgste_j\}_{j\in\NN}$ to the partition $(\GSte,$ $\GNeu)$, where we assume $\GSte$ is non-empty and open, such that all eigenfunctions are mutually orthogonal to each other within the $\Leu^2(\GSte)$ inner product. 
 In the following, the eigenfunctions will be determined uniquely, up to the factor $\pm 1$,  after imposing the normalization condition $\|u_j^0\|_{\Leu^2(\GSte)}=1$.
%

We define $\GDel\subset\del\Om$ as a small boundary interval of length $2\eps$, such that $\GDel\subset\GSte$, and $\GSte\setminus\overline{\GDel}$ is open and non-empty. We define $\{ u_j^\eps \}_{j\in\NN}$ to be the eigenfunctions to the Steklov-Neumann eigenvalues $\{\lam_j^\eps\}_{j\in\NN}$ corresponding to the partition $(\GSte\setminus\overline{\GDel}, \GNeu\cup\GDel)$.

\begin{lemma}\label{lemma:lam+g}
	Let $\GSte, \GNeu, \GDel, \{\lam_j^0\}_{j\in\NN}$ be described as above and let $\eps>0$ be small enough. Let $g\in\Leu^2(\GSte)$, let $Q\subset\CC$ be compact in $\CC$, such that $\lam_j^0\not\in Q$ for all $j\in\NN$. Then the following assertions hold.
	\begin{itemize}
	\item[(i)]	There exists $\eps_0>0$ such that for all $\eps<\eps_0$ and $\lam\in Q$
 there exists a unique solution $U^\eps$ to the boundary value problem
 		\begin{align}\label{pde:lemma:Ueps}
		\left\{ 
		\begin{aligned}
			\Laplace      U^\eps	&= 0 \quad &&\text{in} \; &&\Om\,, \\
		 	\del_{\nu}    U^\eps      	&= \lam \,  U^\eps    + g \quad &&\text{on} \; &&\GSte\setminus\GDel \,,\\
		 	\del_{\nu}    U^\eps		&= 0 \quad &&\text{on} \; &&\GNeu\cup\GDel \,,
		\end{aligned}
		\right.
		\end{align}
		and, for a constant $C$ independent of $\eps$,  it holds that
		\begin{align}\label{inequ:lemma:Ueps}
			\NORM{U^\eps}_{\mathrm{H}^1(\Om)}\leq C\,\NORM{g}_{\Leu^2(\GSte)}\,.
		\end{align}
 	\item[(ii)]	Let $U^0$ be the solution to (\ref{pde:lemma:Ueps}) with $\GDel=\varnothing$, then we have that
 		\begin{align}\label{convergence:lemma:Ueps-U0}
 			\NORM{U^\eps-U^0}_{\mathrm{H}^1(\Om)}\xrightarrow{\eps\rightarrow 0}0\,.
 		\end{align}
	\end{itemize}
\end{lemma}

\begin{proof}
	Let us prove item \textit{(i)}. For (\ref{pde:lemma:Ueps}) the Fredholm alternative holds. Therefore, in order to prove the existence and uniqueness of a solution, it suffices to derive (\ref{inequ:lemma:Ueps}). We verify this estimate by contradiction. Assume that Estimate (\ref{inequ:lemma:Ueps}) fails, that is, there exists a sequence $\eps_k\rightarrow 0$ as $k\rightarrow\infty$, $g_k\in\Leu^2(\GSte)$, uniformly bounded with respect to the $\Leu^2(\GSte)$-norm, and $\lam_k\in Q$ such that the solution satisfies the reverse inequality,
	\begin{align}\label{lemma:proof:contra inequ}
		\NORM{U^{\eps_k}}_{\mathrm{H}^1(\Om)}>k\,\NORM{g_k}_{\Leu^2(\GSte)}\,.
	\end{align}
	Consider that when the pair $(U^\eps\,, g_k)$ satisfies (\ref{pde:lemma:Ueps}) then so does the pair $(\frac{U^\eps}{\NORM{U^\eps}_{\Leu^2(\GSte)}}\,, \frac{g_k}{\NORM{U^\eps}_{\Leu^2(\GSte)}})$ satisfy (\ref{pde:lemma:Ueps}), and for both cases the Estimate (\ref{inequ:lemma:Ueps}) has exactly the same form. Thus we redefine $U^\eps$ as $\frac{U^\eps}{\NORM{U^\eps}_{\Leu^2(\GSte)}}$ and $g_k$ as $\frac{g_k}{\NORM{U^\eps}_{\Leu^2(\GSte)}}$. We notice that $g_k$ is still uniformly bounded because $\NORM{U^\eps}_{\Leu^2(\GSte)}$ cannot converge to 0.
	
	
	
	From \cite[Theorem 6.1]{SteklovExistence} we have that
	\begin{align*}
		\NORM{U^\eps}_{\Leu^2(\Om)}^2
		\leq C_\alpha
			\left(
				\int_\Om |\nabla U^\eps|^2 \;\intd x + \int_{\GSte\setminus\GDel} |U^\eps|^2\;\intd \sigma
			\right),
	\end{align*}		
	for a constant  $C_\alpha>0$ independent of $\eps$.
	
	Then we see using Green's first identity, the Cauchy-Schwarz inequality and the trace theorem that
	\begin{align*}
		\NORM{U^\eps}_{\mathrm{H}^1(\Om)}^2 
			&\leq\tilde{C}_\alpha\left(
					|\lam+1|\,\NORM{U^\eps}^2_{\Leu^2(\GSte)}
					+ \NORM{g_k}_{\Leu^2(\GSte)}\NORM{U^\eps}_{\Leu^2(\GSte)}
				\right)\\
			&\leq \tilde{C}_\alpha\,C_{\mathrm{trace}}\left(
					C_Q\,\NORM{U^\eps}_{\Leu^2(\GSte)}\NORM{U^\eps}_{\mathrm{H}^1(\Om)}
					+ \NORM{g_k}_{\Leu^2(\GSte)}\NORM{U^\eps}_{\mathrm{H}^1(\Om)}
				\right)\,.
	\end{align*}
	From the redefinition we have that $\NORM{U^\eps}_{\Leu^2(\GSte)}=1$ and with the fact that $g_k$ is uniformly bounded, we infer that $\NORM{U^\eps}_{\mathrm{H}^1(\Om)}$ is bounded above, independently of $\eps$. With Inequality (\ref{lemma:proof:contra inequ}) we then have
	\begin{align*}
		\NORM{g_k}_{\Leu^2(\GSte)}\leq \frac{C}{k}\,, 
	\end{align*}
	for a constant  $C>0$ independent of $k$.
	Using the compactness of $Q$, the Rellich-Kondrachov theorem, and the Banach-Alaoglu theorem we conclude that there exists a subsequence to $(\eps_k)_{k\in\NN}$, which by abuse of notation we call again $(\eps_k)_{k\in\NN}$, such that
	\begin{align*}
		\lam_k\xrightarrow{k\rightarrow\infty}&\lam^\star\in Q\,,\\
		U^{\eps_k}\xrightarrow{k\rightarrow\infty}& U^\star\;\text{strongly in $\Leu^2(\Om)$ and weakly in $\mathrm{H}^1(\Om)$}\,.
	\end{align*}
	By the redefinition $\NORM{U^\eps}_{\Leu^2(\GSte)}=1$, we have that $U^\star\neq 0$.
	Consider that the weak formulation of  (\ref{pde:lemma:Ueps}) reads that for all $v\in C^\infty(\overline{\Om})$ we have
	\begin{align}\label{proof-lemmaUeps:itemI:varForm}
		\int_\Om \nabla U^\eps\cdot \overline{\nabla v} \;\intd x - \lam\,\int_{\GSte\setminus\GDel} U^\eps \, \overline{v}\;\intd\sigma 
			=\int_{\GSte\setminus\GDel} g \, \overline{v}\;\intd\sigma \,.
	\end{align}
	We consider this for the sequence $(\eps_k)_{k\in\NN}$, and then apply the limits for $k\rightarrow\infty$ on both sides. Using the dominated convergence theorem we can pull the limit inside and infer that
	\begin{align*}
		\int_\Om \nabla U^\star\cdot \overline{\nabla v} \;\intd x - \lam^\star\,\int_{\GSte} U^\star \, \overline{v}\;\intd\sigma 
			= 0\,,
	\end{align*}
	where we used that $\NORM{g_k}_{\Leu^2(\GSte)}\xrightarrow{k\rightarrow\infty}0$, thus $g_k\xrightarrow{k\rightarrow\infty}0$ in $\Leu^2$. 
	
	From this identity and the fact that $U^\star\neq 0$ we can deduce that $\lam^\star\in Q$ is an eigenvalue with eigenfunction $U^\star$ to  (\ref{pde:SteklovHom}). This contradicts the assumption that no eigenvalues are in $Q$. This contradiction lets us conclude that (\ref{inequ:lemma:Ueps}) holds.
	
	Let us prove item \textit{(ii)}. From the proof of item \textit{(i)} we find for every sequence a subsequence denoted $(\eps_k)_{k\in\NN}$ such that $U^{\eps_k}\xrightarrow{k\rightarrow\infty} U^\star$, where the convergence is strongly in $\Leu^2(\Om)$ and weakly in $\mathrm{H}^1(\Om)$. Passing to the limit for $k\rightarrow\infty$ in the variational formulation (\ref{proof-lemmaUeps:itemI:varForm}), we see that $U^\star$ satisfies the same variational formulation as $U^0$, hence $U^\star=U^0$. Due to the arbitrariness of the sequence, it holds for $\eps\rightarrow 0$. 
	
	Now we need to improve the convergence to strong convergence in $\mathrm{H}^1(\Om)$. Using the strong convergence in $\Leu^2(\Om)$, it is enough to show that
	\begin{align*}
		\lim_{\eps\rightarrow 0} \NORM{\nabla (U^\eps-U^0)}^2_{\Leu^2(\Om)}=0 \,.
	\end{align*}
	Using Green's first identity on $\int_\Om |\nabla(U^\eps-U^0)|^2$, as well as on the definition of weak convergence in $\mathrm{H}^1(\Om)$ (and on weak convergence in $\Leu^2(\Om)$), the last equation reads
	\begin{align*}
		\lim_{\eps\rightarrow 0} \;\lam \NORM{U^\eps-U^0}^2_{\Leu^2(\GSte\setminus\GDel)}=0\,.
	\end{align*}
	Assume for a contradiction that the the last identity is not true.  Denote $V^\eps\DEF U^\eps-U^0$. Then $\lim_{\eps \rightarrow 0}V^\eps \not=0$, and $V^\eps$ satisfies  (\ref{pde:SteklovHom}), hence $\lam$ is an eigenvalue, which by assumption it cannot be, establishing the contradiction. Lemma \ref{lemma:lam+g} is then proved.
\end{proof}

The proof of the following proposition is given in \cite[Theorem 1.2]{SteklovWithSmallDirichlet}, using Lemma \ref{lemma:lam+g}.
\begin{proposition}\label{prop:ueps conv to u0}
	Let $\GSte, \GNeu, \GDel, \{\lam_j^0\}_{j\in\NN},\{\lam_j^\eps\}_{j\in\NN},\{u_j^0\}_{j\in\NN},\{u_j^\eps\}_{j\in\NN}$ be described as above and let $\eps>0$ be small enough. Pick $j\in\NN$ and let $\lam_j^0$ have multiplicity $N_j$. Then the following assertions hold:
	\begin{itemize}
	\item[(i)]	The eigenvalue $\lam_j^0$ is the limit of $N_j$ eigenvalues (with multiplicity taken into account) for $\eps\rightarrow 0$ of the perturbed Steklov-Neumann problem (\ref{pde:SteklovHom}).
 	\item[(ii)]	If $\lam_j^\eps,\ldots,\lam_{j+N_j-1}^\eps$ are eigenvalues of the perturbed Steklov-Neumann problem (\ref{pde:SteklovHom}) converging to $\lam_j^0$ and $u_j^\eps,\ldots, u_{j+N}^\eps$ are the corresponding $\Leu^2(\GSte)$-orthonormalized eigenfunctions, then for any sequence $\eps_k\xrightarrow{k\rightarrow \infty}0$ there exists a subsequence $\eps_{k_m}\xrightarrow{m\rightarrow \infty}0$ such that
 	\begin{align*}
 		\NORM{u_{j+i}^{\eps_{k_m}}-u_{j+i}^0}_{\mathrm{H}^1(\Om)}\xrightarrow{m\rightarrow\infty}0\,,\; \text{ for all } i\in\{0,\ldots,N_j-1\}\,,
 	\end{align*}
 	where $u^0_{j},\ldots,u^0_{j+N-1}$ are $\Leu^2(\GSte)$-orthonormalized eigenfunctions associated with the eigenvalue $\lam_j^0$.
	\end{itemize}
\end{proposition}


We note here that due to Weyl's Lemma all eigenfunctions $u_j^0$ are smooth in the interior of the domain. We also readily see that they are continuous at the interior of any line segment. However they might be discontinuous at the endpoints of those line segments.

\begin{theorem}\label{thm:EVAL-Asymp}
	Let $\GSte$, $\GNeu$, $\GDel$, $\{\lam_j^0\}_{j\in\NN}$, $\{\lam_j^\eps\}_{j\in\NN}$, $\{u_j^0\}_{j\in\NN}$, $\{u_j^\eps\}_{j\in\NN}$ be described as above, let $c_\star\in\GDel$ be the center of $\GDel$ and let $\eps>0$ be small enough. Pick $j\in\NN$. Let $\{u_j^0,\ldots,u_{j+N_j-1}^0\}$ be $\Leu^2(\GSte)$-normalized linearly independent eigenfunctions of $\lam_j^0$ and let $\lam^\eps_j$ converge to $\lam^0_j$
	
	If there exists a $i\in\{0,\ldots,N_j-1\}$ such that $u_{j+i}^0(c_\star)\neq 0$ then there exists an $u_j^{\eps_n}$ converging in $\mathrm{H}^1(\Om)$ to an element in the eigenspace of $\lam_j^0$, where
	\begin{align*}
		\lam_j^\eps-\lam_j^0 \, 
			&= \,2\,\eps\, \lam_j^0\sum_{i=0}^{N_j-1} (u_{j+i}^0(c_\star))^2 \,
			+\, \OO(\eps^2)\,,\\
		u_{j}^{\eps_{n}}(\xS)
			&= \frac{\sum_{i=0}^{N_j-1}u_{j+i}^0(\xS)\,u_{j+i}^0(c_\star)}
			{\Big[\sum_{i=0}^{N_j-1}u_{j+i}^0(c_\star)^2\Big]^{\nicefrac{1}{2}}}
			+\OO(\eps_n^1) \quad\text{ for } \xS\in\Om\,.
	\end{align*}
	Let $u_{j,\oplus}^{\eps_n}$ be another perturbed eigenfunction, as described in Proposition \ref{prop:ueps conv to u0}, then $u_{j,\oplus}^{\eps_n}(c_\star)=0$ and the associated eigenvalue $\lam_{j,\oplus}^\eps$ satisfies
	\begin{align*}
		\lam_{j,\oplus}^\eps-\lam_{j}^0
			= o(\eps^2)\,.
	\end{align*}
\end{theorem}

From the proof we can retrieve a formula for the $\OO(\eps^1_n)$ part in $u_j^\eps$ up to order $\OO(\eps^2)$. But we will not use it and for sake of readability it is omitted. 

The structure of the proof is as follows. Let $(\cdot\,,\cdot)_{B}$ denote the $\Leu^2(B)$ inner product over an integrable domain $B$ and $\mathrm{S}_\GSte^{\lam^{\eps}}$ the Green's function to the Steklov-Neumann problem. Using Green's identity we will first readily obtain 
\begin{align*}
		(\lam_j^\eps-\lam_j^0)\, ( u_{j}^0 \,,u_j^\epsilon)_{\GSte\setminus\GDel}
			&= \lam^0_j\, ( u_{j}^0\,, u_j^\epsilon)_\GDel
\end{align*}
and 
\begin{align*}
	u^{\eps}_j(\xS)
			= \lam_j^\eps( \mathrm{S}_\GSte^{\lam^{\eps}}\,, u^{\eps}_j)_{\GDel}\,.
\end{align*}
Using both equations and the decomposition in Theorem \ref{thm:decomp S}, together with the condition $\del_\nu u_j^\eps |_{\GDel}=0$ from the governing equation, we will obtain a formula for $u_j^\eps |_{\GDel}$. From there on we can trace back and obtain a formula for the eigenfunction on the whole domain $\Om$. Using again the above equation we will finally obtain the asymptotic formula for $\lam_j^\eps-\lam_j^0$. 

This also supplies us with a way to find the limit of an perturbed eigenfunction. Using orthogonality between two limits of perturbed eigenfunctions, we readily argue why every other eigenfunction must be zero at $c_\star$. Using the above expression for $\lam_j^\eps-\lam_j^0$ together with the Taylor expansion of the limit, we infer the convergence rate $o(\eps^2)$. 

\begin{proof}
	We denote the sequence $\{\eps_n\}_n$ by $\eps\rightarrow 0$ and the $\Leu^2(\GSte)$-normalized eigenfunctions to $\lam_j$ by $\{u_j^0,\ldots, u_{j+N_j-1}^0\}$.  Using Green's first identity and the governing equation (\ref{pde:SteklovHom}), we have for $i=0,\ldots,N_j-1$ that
	\begin{align*}
		( \nabla u_{j+i}^0 \,, \nabla u_j^\epsilon)_\Om
			&= (\del_{\nu} u_{j+i}^0 \,, u_j^\epsilon)_{\del\Om}
			= \lam^0_j ( u_{j+i}^0 \,, u_j^\epsilon)_\GSte\,,\\
		( \nabla u_{j+i}^0 \,, \nabla u_j^\epsilon)_\Om
			&= ( u_{j+i}^0 \,, \del_{\nu} u_j^\epsilon)_{\del\Om}
			= \lam^\eps_j ( u_{j+i}^0\,, u_j^\epsilon)_\GSte
				- \lam^\eps_j ( u_{j+i}^0 \,,u_j^\epsilon)_\GDel\,.
	\end{align*}
	From this it is easy to see that
	\begin{align}
		(\lam_j^\eps-\lam_j^0)\, ( u_{j+i}^0 \,,u_j^\epsilon)_{\GSte\setminus\GDel}
			&= (-\lam_j^\eps+\lam_j^0)\, ( u_{j+i}^0\,, u_j^\epsilon)_\GDel +\lam_j^\eps( u_{j+i}^0\,, u_j^\epsilon)_\GDel\nonumber \\
			&=
			\lam^0_j\, ( u_{j+i}^0\,, u_j^\epsilon)_\GDel\,.\label{equ:EVALExpansion-without-EFct-Taylor}
	\end{align}
	Using Green's first identity and the decomposition in Theorem \ref{thm:decomp S}, we have
	\begin{align*}
		u^{\eps}_j(\xS)
			&= \lam_j^\eps( \mathrm{S}_\GSte^{\lam^{\eps}}\,, u^{\eps}_j)_{\GDel}\\
			&=  \lam_j^\eps( \Gamma^0\,, u^{\eps}_j)_{\GDel}
				+ \lam_j^\eps \sum_{i=0}^{N_j-1} \frac{u_{j+i}^0(\xS)\, (u_{j+i}^0\,, u^{\eps}_j)_{\GDel}}{\lam_j^\eps-\lam_j^0}
				+ \lam_j^\eps\,(\mathrm{R}_\GSte^{\lam^\eps}\,, u^{\eps}_j)_{\GDel}\,,
	\end{align*}
	where $\mathrm{R}^{\lam^\eps}_{\GSte}(\xS\,,\cdot)\in\mathrm{H}^1(\Om)\cap C^\infty(\Om)$ and it is analytic in a neighborhood of $\lam_j^\eps$, according to Theorem \ref{thm:decomp S}. 
	
	From (\ref{equ:EVALExpansion-without-EFct-Taylor}) follows $\frac{(u_{j+i}^0\,, u^{\eps}_j)_{\GDel}}{\lam_j^\eps-\lam_j^0}=\frac{(u_{j+i}^0\,, u^{\eps}_j)_{\GSte\setminus\GDel}}{\lam_j^0}$, thus we can infer that
	\begin{align}\label{equ:ueps-general-formula}
		u^{\eps}_j(\xS)
			= \lam_j^\eps( \Gamma^0\,, u^{\eps}_j)_{\GDel}
				+ \frac{\lam_j^\eps}{\lam_j^0} \sum_{i=0}^{N_j-1} u_{j+i}^0(\xS)\, (u_{j+i}^0\,, u^{\eps}_j)_{\GSte\setminus\GDel}
				+ \lam_j^\eps\,(\mathrm{R}_\GSte^{\lam^\eps}\,, u^{\eps}_j)_{\GDel}\,.
	\end{align}
	
	{Without loss of generality, we assume that  $\GDel=(-\eps,\eps) \times \{0\}$ and that the outside normal on $\GDel$ is $(0,1)^\TransT$, where the superscript $\TransT$ denotes the transpose.}
	
	For $y\in\GDel$, we define $(\tau,0)^\TransT=y$. By definition, $\del_\nu u_j^\eps(y) = 0$ for $y\in\GDel$, which yields with the last equation that
	\begin{align*}
		\lim_{h\searrow 0}
		\frac{2\lam_j^\eps}{2\pi}\inteps &\frac{h\,u_j^\eps(t)}{h^2+(t-\tau)^2}\intd t\\
			& = \lam_j^\eps \sum_{i=0}^{N_j-1} u_{j+i}^0(y)\, (u_{j+i}^0\,, u^{\eps}_j)_{\GSte\setminus\GDel}
				+ \lam_j^\eps\,(\del_\nu \mathrm{R}_\GSte^{\lam^\eps}\,, u^{\eps}_j)_{\GDel}\,.
	\end{align*}
	Here we considered that we in fact placed the singularity at the boundary, this means that the associated Dirac measure is halved, thus we have $\mathrm{S}_\GSte^{\lam^{\eps}} = 2\Gamma^0 + \mathrm{R}_\GSte^{\lam^\eps}$. Furthermore, we can pull the normal derivative inside the $(\mathrm{R}_\GSte^{\lam^\eps}\,, u^{\eps}_j)_{\GDel}$, which follows by using Green's identity and the dominated convergence theorem.
	Using partial integration we  readily see that 
	\begin{align*}
		\lim_{h\searrow 0}\int_{-\eps}^\eps \frac{h}{h^2+(\tau-t)^2}\,u_j^\eps(t)\intd t = \pi u_j^\eps(\tau).
	\end{align*}
	This leads us to
	\begin{align}\label{proof:thm O(eps) conv:uijeps at GDel}
		u_j^\eps(t) 
			= \sum_{i=0}^{N_j-1} u_{j+i}^0(y)\, (u_{j+i}^0\,, u^{\eps}_j)_{\GSte\setminus\GDel}
				+ (\del_\nu \mathrm{R}_\GSte^{\lam^\eps}\,, u^{\eps}_j)_{\GDel}\,.
	\end{align}
	Using  (\ref{equ:EVALExpansion-without-EFct-Taylor}), we see that for all $i\in\{0,\ldots,N_j-1\}$
	\begin{align}\label{proof:thm O(eps) conv:(uji0,uje)=frac(uj0,uje)} 
		(u_{j+i}^0\,,u_j^\eps)_{\GSte\setminus\GDel}
			=\frac{(u_{j+i}^0\,,u_j^\eps)_{\GDel}}
			{(u_{j}^0\,,u_j^\eps)_{\GDel}}
			(u_{j}^0\,,u_j^\eps)_{\GSte\setminus\GDel}\,,
	\end{align}
	where
	\begin{align}\label{proof:thm O(eps) conv:(uji0,uje)/(uj0,uje)=(f+o)} 
		\frac{(u_{j+i}^0\,,u_j^\eps)_{\GDel}}
		{(u_{j}^0\,,u_j^\eps)_{\GDel}}
			=\frac{u_{j+i}^0(c_\star)}{u_{j}^0(c_\star)}+\OO(\eps^1)\,,
	\end{align}
	which we readily infer from (\ref{proof:thm O(eps) conv:uijeps at GDel}), and where we used the assumption $u_{j+i}^0(c_\star)\neq 0$, where we assumed $i=0$ without loss of generality.
	From (\ref{equ:ueps-general-formula}) we infer for $\xS\in\overline{\Om}\setminus\overline{\GDel}$ that
	\begin{align*}
		u_j^\eps(\xS) 
			=& \frac{\lam_j^\eps}{\lam_j^0}\,(u_{j}^0\,,u_j^\eps)_{\GSte\setminus\GDel}
			\sum_{i=0}^{N_j-1}
				u_{j+i}^0(\xS) \frac{u_{j+i}^0(c_\star)}{u_{j}^0(c_\star)}
			\,+\OO(\eps^1)\,.
	\end{align*}
	Let us determine $(u_{j}^0\,,u_j^\eps)_{\GSte\setminus\GDel}$. From $(u_{j}^\eps,u_{j}^\eps)_{\GSte\setminus\GDel}=1$ and (\ref{proof:thm O(eps) conv:(uji0,uje)=frac(uj0,uje)}), it follows readily that
	\begin{align*}
		(u_{j}^0\,,u_j^\eps)_{\GSte\setminus\GDel}^2
			= \bigg[
				\frac{\lam_j^\eps}{\lam_j^0}\sum_{i=0}^{N_j-1}\frac{u_{j+i}^0(c_\star)^2}{u_{j}^0(c_\star)^2}
			\bigg]^{-1}
			\bigg(
				1
				+\OO(\eps^1)
			\bigg)\,.
	\end{align*}
	We readily see that $\frac{\lam_j^\eps}{\lam_j^0}=1+\OO(\eps)$, hence, we infer the formula for $u_j^\eps(\xS)$ in Theorem \ref{thm:EVAL-Asymp}. By using (\ref{equ:EVALExpansion-without-EFct-Taylor}) we obtain
	\begin{align*}
		\lam_j^\eps-\lam_j^0\,
			= \frac{\lam_j^0}{(u_j^0\,, u_j^\eps)_{\GSte\setminus\GDel}}
			\sum_{i=0}^{N_j-1}\Big[
				(u_{j}^0,u_{j+i}^0)_\GDel \frac{u_{j+i}^0(c_\star)}{u_{j}^0(c_\star)}(u_{j}^0\,, u_j^\eps)_{\GSte\setminus\GDel}
			\Big]+\OO(\eps^2)\,.
	\end{align*}
	This leads us to the formula for $\lam_j^\eps-\lam_j^0$ in Theorem \ref{thm:EVAL-Asymp}.
	
	Let us consider other perturbed eigenfunctions. We claim that every other perturbed eigenfunctions are zero at $c_\star$. To this end, in the case that not all eigenfunctions are zero at $c_\star$, we can pick through the Gram-Schmidt orthogonalization an $\Leu^2(\GSte)$-orthonormal basis $\{u_j^0,\ldots,u_{j+N_j-1}^0\}$ of $\lam_j$ such that $u_{j+i}^0(c_\star)=0$ for all $i\in\{1,\ldots,N_j-1\}$, and $u_{j}^0(c_\star)\neq 0$. Assume that two perturbed eigenfunctions $u_{j,0}^\eps,u_{j,\oplus}^\eps$ converge to the eigenspace of $\lam_j^0$. Then by the previously derived formula, one of those, say $u_{j,0}^\eps$, converges to
	\begin{align*}
		\frac{\sum_{i=0}^{N_j-1}u_{j+i}^0(\xS)\,u_{j+i}^0(c_\star)}
			{\Big[\sum_{i=0}^{N_j-1}u_{j+i}^0(c_\star)^2\Big]^{\nicefrac{1}{2}}}
			= u_{j}^0(\xS)\frac{u_{j}^0(c_\star)}
			{|u_{j}^0(c_\star)|}\,.
	\end{align*}
	Since the limits of the two perturbed eigenfunctions are orthogonal to each other, we infer that $u_{j,\oplus}^\eps$ converges to the span of $\{u_{j+1}^0,\ldots,u_{j+N_j-1}^0\}=(u_{j}^0)^{\perp}$. Hence $\lim_{\eps\rightarrow 0} u_{j,\oplus}^\eps(c_\star)=0$. Let $\lam_{j,\oplus}^\eps$ be the associated eigenvalue to $u_{j,\oplus}^\eps$ and $u_{j,\oplus}^0$ its limit. With (\ref{equ:EVALExpansion-without-EFct-Taylor}) we infer that
	\begin{align*}
		(\lam_j^\eps-\lam_j^0)
			= \lam^0_j\, \frac{( u_{j,\oplus}^0\,, u_{j,\oplus}^\epsilon)_\GDel}{( u_{j,\oplus}^0 \,,u_{j,\oplus}^\epsilon)_{\GSte\setminus\GDel}}
			= \frac{\OO(\eps^2)\,o(1)}{1+o(1)}\,.
	\end{align*}
	This concludes the proof of Theorem \ref{thm:EVAL-Asymp}.
\end{proof}

\newcommand{\SteOm}{{S_{\del\Om}^\lam}}
\newcommand{\SteSte}{{S_{\GSte}^\lam}}
\newcommand{\SteM}{{S_{\mathrm{M}}^\lam}}
\newcommand{\UGSl}{{\mathrm{U}_\GSte^\lam}}
\newcommand{\ZGS}{{\mathrm{Z}_\GSte}}
\section{Asymptotics for the Steklov-Neumann Function}\label{Sec:STFct Asymp}
Let $\SteSte(\xS\,,\cdot)$ be the Green's function to the Steklov-Neumann problem (\ref{pde:SteklovFunda}) to the partition $(\GSte,\GNeu)$ of $\del\Om$, where we assume that $\GSte$ is non-empty and open. Let $\{ u_j \}_{j\in\NN}$ be the normalized Steklov-Neumann eigenfunctions corresponding to the Steklov-Neumann eigenvalues $\{\lam_j\}_{j\in\NN}$ to the partition $(\GSte,\GNeu)$, such that all eigenfunctions are mutually orthogonal to each other within the $\Leu^2(\GSte)$ inner product.

\begin{theorem}\label{thm:decomp S}
	Let $\GSte,\GNeu, \{ u_j \}_{j\in\NN}, \{\lam_j\}_{j\in\NN}$ be as described above. Let $\lam\geq 0$ such that $\lam \neq \lam_j$ for all $j\in\NN$. For all $\xS\in\Om$, the Steklov-Neumann function $\SteSte(\xS\,,\cdot)$, given by (\ref{pde:SteklovFunda}), exists and $\SteSte(\xS\,,\cdot)\in\Leu^2(\Om)$.
	
	Furthermore pick $j\in\NN$, then for $y\in\Om$, $y\neq \xS$ and for $\lam$ close enough to $\lam_j$, where $\lam$ has multiplicity $N_j$, that is $\lam_j=\ldots=\lam_{j+(N_j-1)}$, we have that
	\begin{align*}
		\SteSte(\xS\,, y) 
			= \frac{1}{2\pi}\log(|x-y|)+\sum_{i=0}^{N_j-1} \frac{u_{j+i}(\xS)\,u_{j+i}(y)}{\lam-\lam_j}+\mathrm{R}_{\GSte}^\lam(\xS\,, y)\,.
	\end{align*}
	Moreover, $\mathrm{R}^\lam_{\GSte}(\xS\,,\cdot)\in\mathrm{H}^1(\Om)\cap C^\infty(\Om)$ and it is analytic in a neighborhood of $\lam_j$.
\end{theorem}

We remark here that we expect that the term $\mathrm{R}_{\GSte}^\lam$ blows up whenever the length of $\GSte$ goes to zero. The function $\frac{1}{2\pi}\log(|x-y|)$ is the Green's function to the free-space Laplace equation, we refer to \cite[Section 2.2]{LPTSA} and \cite[Section 2.3]{MCMP}.

\begin{proof}
	We use the decomposition 
	\begin{align}\label{proof:thm decompS: decomp}
		\SteSte(\xS\,, y)
			= \Gamma^0(\xS\,, y) + \ZGS(\xS\,,y) + \UGSl(\xS\,,y)\,,
	\end{align}
	where $\Gamma^0(\xS\,, y)\DEF\frac{1}{2\pi}\log(|x-y|)\in \Leu^2(\Om)$, and the functions $\ZGS(\xS\,,\cdot)\,, $ $\,\UGSl(\xS\,, \cdot)$ are solutions to 
	\begin{align*}
		\left\{ 
		\begin{aligned}
			\Laplace      \ZGS	&= 0 					\quad &&\text{in} \; &&\Om\,, \\
			\ZGS 				&= 0 				    \quad &&\text{on} \; &&\GSte\,, \\
		 	\del_{\nu}    \ZGS  &= -\del_{\nu} \Gamma^0 \quad &&\text{on} \; &&\GNeu \,,
		\end{aligned}
		\right.
		\quad
		\left\{ 
		\begin{aligned}
			\Laplace      \UGSl	&= 0 \quad &&\text{in} \; &&\Om\,, \\
			\del_{\nu}    \UGSl-\lam\,\UGSl		&= f \quad &&\text{on} \; &&\GSte\,, \\
		 	\del_{\nu}    \UGSl &= 0 \quad &&\text{on} \; &&\GNeu \,,
		\end{aligned}
		\right.
	\end{align*}
	where 
	$$
		f(\xS\,, \cdot)\DEF-(\del_\nu\Gamma^0(\xS\,, \cdot)+\del_\nu \ZGS(\xS\,, \cdot)-\lam\,\Gamma^0(\xS\,, \cdot))\in \Leu^2(\GSte)\,.
	$$ 
	From \cite[Theorem 4.10]{McLeanEllitpicSystems}, and the fact that the first mixed Dirichlet - Neumann eigenvalue is not zero \cite{ourzarembapaper}, as long as $\GSte$ is not empty, we have that $\ZGS(\xS\,,\cdot)\in\mathrm{H}^1(\Om)$ exists. From Lemma \ref{lemma:lam+g}, we see that $\UGSl(\xS\,,\cdot)\in\mathrm{H}^1(\Om)$ exists. We infer that $\SteSte(\xS\,,\cdot)$ exists and that $\SteSte(\xS\,,\cdot)\in\Leu^2(\Om)$. 
	
	With a small modification in the proof of \cite[Theorem 7.3]{SteklovExistence}, respectively \cite[Theorem 5.3]{SteklovExistence}, we see that $\{u_j\}_{j\in\NN}$ is a maximal orthonormal set to 
	$$
		W_{\del}\DEF\{v\in\mathrm{H}^1(\Om) \,|\, v \text{ weak solution to }\Laplace v = 0 \text{ and  }\del_\nu v\,|_{\GNeu}=0\}\,,
	$$
	within the $\mathrm{H}^1(\Om)$ inner product $\mathscr{A}(\cdot\,,\cdot)$ given in \cite[Equation (6.5)]{SteklovExistence}. Consider that $\UGSl\in W_{\del}$, thus we can write
	\begin{align}\label{proof:thm decompS: Sum U}
		\UGSl(\xS\,,y) 
			= \sum_{j=1}^\infty (\UGSl(\xS\,,\cdot)\,, u_j(\cdot))_{\GSte}\,u_j(y)\,,
	\end{align}
	where $(\cdot\,,\cdot)_{B}$ denotes the $\Leu^2(B)$ inner product over an integrable domain $B$. Using Green's first identity and the governing equations multiple times, we obtain that
	\begin{align*}
		-\lam\,(\UGSl\,, u_j)_\GSte
			&= (-\del\UGSl\,, u_j)_\GSte +(-\del\Gamma^0+\lam\,\Gamma^0\,, u_j)_\GSte+(-\del\ZGS\,,u_j)_\GSte\,\\
			&= -\lam_j\,(\UGSl\,, u_j)_\GSte + (-\del\Gamma^0+\lam\,\Gamma^0\,, u_j)_\GSte+0-(\del\Gamma^0\,, u_j)_\GNeu\\
			&= -\lam_j\,(\UGSl\,, u_j)_\GSte  -u_j(\xS)-(\lam_j-\lam)(\Gamma^0\,, u_j)_\GSte\,.
	\end{align*}
	After rearranging the terms we obtain that
	\begin{align}\label{proof:thm decompS: (U,uj)}
		(\UGSl\,, u_j)_\GSte
			= \frac{u_j(\xS)}{(\lam-\lam_j)}-(\Gamma^0\,, u_j)_\GSte\,.
	\end{align}
	With the definition 
	\begin{align*}
		\mathrm{R}^\lam_{\GSte}(\xS\,, y)
			\DEF \ZGS(\xS\,,y) + \UGSl(\xS\,,y)
			-\sum_{i=0}^{N_j-1} \frac{u_{j+i}(\xS)\,u_{j+i}(y)}{\lam-\lam_j}\,,
	\end{align*}
	and Weyl's lemma, we conclude that Theorem \ref{thm:decomp S} holds.
\end{proof}

Next we search for an asymptotic formula for $\SteSte$, when we insert a small Neumann condition in the boundary. To this end, we define $\GDel$ as a small boundary interval of length $2\eps$ with center $c_\star\in\RR^2$, such that ${\GDel}\subset\GSte$. Then we define $\SteM(\xS\,,\cdot)$ to be the Green's function to the Steklov-Neumann problem (\ref{pde:SteklovFunda}) to the partition $(\GSte\setminus\overline{\GDel}, \GNeu\cup\GDel)$, where $\lam$ is not a Steklov-Neumann eigenvalue to both partitions.

\begin{theorem}\label{thm:FundaSteklovAsympt}
	Let $\GSte,\GNeu, \GDel$ be as described above and let $\eps>0$ be small enough. Let $\lam\geq 0$ not be a Steklov-Neumann eigenvalue to either of the two partitions $(\GSte,\GNeu)$ and $(\GSte\setminus\overline{\GDel}, \GNeu\cup\GDel)$. Then we have for all $\xS\in \Om, y\in\Om\setminus\{\xS\}$ that
	\begin{align*}
		\SteM(\xS\,,y) = \SteSte(\xS\,, y) + 2\,\lam\,\eps\,\SteSte(\xS\,, c_\star)\,\SteSte(y\,, c_\star)+\OO(\eps^2)\,.
	\end{align*}
	Let $y\in\GDel$ then
	\begin{align*}
		\SteM(\xS\,,y) =  \SteSte(\xS\,, y)+\OO(\eps)\,.
	\end{align*}
\end{theorem}

Let us consider the structure of our proof. Using Green's identity we establish the identity
\begin{align*}
	-\lam\,\SteSte(\xS\,, y)
		= \lam \del_{\nu_y}
			(\SteSte({y},\cdot)\,,\SteM(\xS\,, \cdot))_\GDel\,,
\end{align*}
for $y\in\GDel$, where $(\cdot\,,\cdot)_{B}$ denotes the $\Leu^2(B)$ inner product over an integrable domain $B$. Using the decomposition established in Theorem \ref{thm:decomp S}, we can recover a formula for $\SteM(\xS\,, y)$ for $y\in\GDel$. Then we can trace back, that is using Green's identity one more time, to find the equation in Theorem \ref{thm:FundaSteklovAsympt} for $y\in\Om$.

\begin{proof}
	Without loss of generality, we assume that $\GDel$ is a straight line segment on $\del\Om$. By rotating and translating $\Om$, we assume that $\GDel=(-\eps,\eps) \times \{0\}$ and that the outside normal on $\GDel$ is $(0,1)^\TransT$. We remind here that this is possible since we assume that $\del\Om$ is smooth.
	
	With Green's second identity we obtain for $y\in\Om$
	\begin{align*}
		\SteM(\xS\,, y) 
			&=  ( \Laplace_z \SteSte(y\,,\cdot)\,, \SteM(\xS\,,\cdot))_\Om\\
			&=	\SteSte(y\,,\xS)+( \del_{\nu_z}\SteSte(y\,,\cdot)\,,\SteM(\xS\,,\cdot))_\GDel\,.
	\end{align*}
	We readily see that $\SteSte(\xS\,,y)=\SteSte(y\,,\xS)$, using the same argument as the one in the last equation. We define $v_{\xS}(y)\DEF \SteM(\xS\,, y)-\SteSte(\xS\,, y)$. This leads us to
	\begin{align}
		\left\{ 
			\begin{aligned}
		 		\del_{\nu_y}   v_{\xS}(y)	
		 			&= -\lam\,\SteSte(\xS\,,y) \quad && \text{for}\;y\in\GDel\,, \\
		 		v_{\xS}(y)      	
		 			&= (\del_{\nu_z}\SteSte(y\,,\cdot)\,,\SteM(\xS\,,\cdot))_\GDel	 \quad && \text{for}\;y\in\Om\,. \\
			\end{aligned}
		\right. \label{equ:v=int()}
	\end{align}
	Combining both statements, we have for $y\in\GDel$
	\begin{align*}
		-\lam\,\SteSte(\xS\,, y)
			= \nu_y\cdot\lim_{\substack{\hat{y}\rightarrow y\\\hat{y}\in\Om}}
			\nabla_{\!\hat{y}}
				\,\lam( \SteSte(\hat{y},\cdot)\,,\SteM(\xS\,, \cdot))_\GDel\,.
	\end{align*}
	Next we use the decomposition $\SteSte(y,z)=2\Gamma^0(y,z)+\mathrm{R}^\lam_\GSte(y,z)$, where $\Gamma^0(y,z)=\frac{1}{2\pi}\log(|y-z|)$ denotes the fundamental solution to the Laplace equation $\Laplace \Gamma^0(y\,,\cdot)=\delta_y(\cdot)$, and $\mathrm{R}^\lam_\GSte(y,\cdot)\subset\mathcal{C}^\infty({\Om})\cap\mathrm{H}^1(\Om)$ denotes the remaining function from Theorem \ref{thm:decomp S}. The factor 2 emerges due to the fact that $z\in\GDel$ is fixed and we consider $y\rightarrow \GDel$, this means that the singularity is at the boundary which halves the Dirac measure. 
	Hence using that $\GDel$ is flat and located at the horizontal axis, centered at the origin, we can rewrite the right-hand side of the last equation as
	\begin{multline}
		\lam\,\nu_y\cdot\lim_{\substack{\hat{y}\rightarrow y\\\hat{y}\in\Om}}
			\nabla_{\!\hat{y}}\left(
				( 2\Gamma^0(\hat{y},\cdot)\,,\SteM(\xS\,, \cdot))_\GDel
				+( \mathrm{R}^\lam_\GSte(\hat{y},\cdot)\,,\SteM(\xS\,, \cdot))_\GDel\,
			\right)	=\\
		\frac{\lam}{2\pi}\lim_{h\searrow 0}
				\int_{\GDel} \frac{-2h}{h^2+(z_1-y_1)^2}\,\SteM(\xS\,, z)\intd \sigma_z
				+\lam ( \del_{\nu_y}\mathrm{R}^\lam_\GSte(y,\cdot)\,,\SteM(\xS\,, \cdot))_\GDel\,.
	\end{multline}
	We can pull the normal derivative inside the term $(\mathrm{R}_\GSte^{\lam}\,, \SteM(\xS\,, z))_{\GDel}$, which follows readily using Green's identity and then the dominated convergence theorem. 
	Using that $\lim_{h\searrow 0}\int_{-\eps}^\eps \frac{h}{h^2+(\tau-t)^2}\,\mu(t)\intd t=\pi \mu(\tau)$, we infer
	\begin{align*}
		\SteM(\xS\,,y)
		=
			\SteSte(\xS\,, y)
			+\OO(\eps)\,.
	\end{align*}	

	Using (\ref{equ:v=int()}) we obtain for $y\in\Om$ that
	\begin{align*}
		\SteM(\xS\,,y) = \SteSte(\xS\,, y) +\lam\big(\SteSte(\xS\,, \cdot)\,,\SteSte(y\,, \cdot)\big)_\GDel+\OO(\eps^2)\,,
	\end{align*}
	from which Theorem \ref{thm:FundaSteklovAsympt} follows.
\end{proof}
\makeatletter
\newenvironment{breakablealgorithm}
  {
   \begin{center}
     \refstepcounter{algorithm}
     \hrule height.8pt depth0pt \kern2pt
     \renewcommand{\caption}[2][\relax]{
       {\raggedright\textbf{\ALG@name~\thealgorithm} ##2\par}%
       \ifx\relax##1\relax 
         \addcontentsline{loa}{algorithm}{\protect\numberline{\thealgorithm}##2}%
       \else 
         \addcontentsline{loa}{algorithm}{\protect\numberline{\thealgorithm}##1}%
       \fi
       \kern2pt\hrule\kern2pt
     }
  }{
     \kern2pt\hrule\relax
   \end{center}
  }
\makeatother

\section{The Algorithm}\label{Ch:Algorithm}

We now present an algorithm to which, by insertion of Neumann boundary pieces,  maximizes the magnitude of the Green's function corresponding to a mixed Steklov-Neumann problem, for a Steklov parameter close to a given target.

Our algorithm starts with a full Steklov boundary condition. Then it changes a small boundary interval into a Neumann boundary condition such that the value of the Steklov-Neumann function $\mathrm{S}^\lam(\xS, y)$ increases, where $\xS$ is the signal transmitting point in the domain $\Om$ and $y$ is the receiving point by using the asymptotic formulae in Theorems \ref{thm:EVAL-Asymp} and \ref{thm:FundaSteklovAsympt}. According to these two theorems, by changing a boundary part from the Steklov boundary condition to the Neumann one, the associated eigenvalue $\lambda_j^{\GDir}$ increases. Thus the idea is to expand the Neumann boundary in such extent that we eventually hit the desired eigenvalue $\lam_\star$. This is not possible if the desired eigenvalue $\lam_\star$ is smaller than the first non-zero Steklov eigenvalue $\lam_2^{\del\Om}$, since the eigenvalue $0$ cannot increase and all other eigenvalues cannot decrease. {However, from our numerical experiments,  it seems to hold true for all other positive $\lam_\star$.} 

\begin{breakablealgorithm}
\caption{Finding an intensity maximizing partition of the boundary}\label{Algorithm}
\hspace*{\algorithmicindent} \textbf{Input:} $\xS\in\Om$, $y\in\Om$, $y\neq \xS$, $\lam_\star>0$, $C_{\mathrm{tol}}>0$. \\
\textbf{Require:} $\lam_\star\geq\lam_2^{\del\Om}$ and $C_{\mathrm{tol}}$ is big enough. 
\begin{algorithmic}[1] 
\State Find the next lower Steklov eigenvalue $\lam$ to $\lam_\star$.
\State Compute the value $\mathrm{S}^\lam_{\del\Om}(\xS,y)$ and the Steklov functions $\mathrm{S}^\lam_{\del\Om}(\xS,\cdot)$, $\mathrm{S}^\lam_{\del\Om}(y,\cdot)$ associated to the partition $(\del\Om, \varnothing)$ at the boundary.

\If{$\mathrm{S}^\lam_{\del\Om}(\xS,y)\geq 0$}
	\State Let $\mathsf{L}\in{\del\Om}$ be the location of a global maxima of the function $\del\Om\ni z\mapsto\big(\mathrm{S}^\lam_{\del\Om}(\xS,z) \cdot \mathrm{S}^\lam_{\del\Om}(y,z)\big)\in\RR$. 
\Else
	\State Let $\mathsf{L}\in{\del\Om}$ be the location of a global minima of the function $\del\Om\ni z\mapsto\big(\mathrm{S}^\lam_{\del\Om}(\xS,z) \cdot \mathrm{S}^\lam_{\del\Om}(y,z)\big)\in\RR$. 
\EndIf
\State Set $\eps, \eps_\Delta\DEF 0$, $f\DEF 1$.
\State Set $\GDel, \GDel^0\DEF\{\mathsf{L}\}$, $\lam_0\DEF\lam$.
\While{$\textrm{True}$}
	\State Set $\{u_i\}_{i=0}^{N-1}$ to be orthonormalized eigenfunctions to $\lam_0$ and $(\del\Om\setminus{\GDel^0},\;\GDel^0)$.
	\State Set $\epsilon_\Delta = f\cdot\frac{1}{2}\frac{\lam_\star-\lam_0}{\lam_0\sum_{i=0}^{N-1}u_{i}(\mathsf{L})^2}$.
	\State Extend $\GDel$ on both sides by the length $\eps_\Delta$.
	\State Compute the new eigenvalue $\lam$ to the partition $(\del\Om\setminus{\GDel},\;\GDel)$.
	\If{$|\lam-\lam_\star|\leq C_{\mathrm{tol}}$}
		\Return $\GDel$ 
	\ElsIf{$\lam<\lam_\star-C_{\mathrm{tol}}$}
		\State Set $\GDel^0=\GDel$, $\eps_0=\eps_0+\eps_\Delta$, $\lam_0=\lam$.
		\State Set $f=1$.
	\Else
		\State Set $\GDel=\GDel^0$, $\lam=\lam_0$.
		\State Set $f=f\cdot 0.8$.
	\EndIf
\EndWhile
\end{algorithmic} 
\end{breakablealgorithm}

In the following we give an explanation for the above choices.
\begin{itemize}
	\item[Line 2-7:] To increase the value $\mathrm{S}^\lam_{\del\Om}(\xS,y)$ we use in the asymptotic formula given in Theorem \ref{thm:FundaSteklovAsympt}. Depending on the sign of $\mathrm{S}^\lam_{\del\Om}(\xS,y)$, we have to search for the maxima of the function $z\mapsto\big(\mathrm{S}^\lam_{\del\Om}(\xS,z) \cdot \mathrm{S}^\lam_{\del\Om}(y,z)\big)$ or its minima;
	\item[Line 10:] In this \textit{while}-loop, we change a boundary interval with center $\mathsf{L}$ and length $2\eps$ into a Neumann Boundary condition. Here, $\eps$ is determined by Theorem \ref{thm:EVAL-Asymp}, that is, $\lam_\star \approx \lam + 2\eps\lam\sum_{i=0}^{N-1} u_i(\mathsf{L})^2$. The factor $f$ dampens the choice of $\eps_\Delta$ and is chosen to be equal to $1$ at the beginning of the algorithm. If the new eigenvalue $\lam$ satisfies $|\lam-\lam_\star|<C_{\mathrm{tol}}$, then we end the loop; if $\lam<\lam_\star-C_{\mathrm{tol}}$, then we restart the loop with the new values, and in the remaining case we decrease $f$. We choose here $f=f\cdot 0.8$ for simplicity, but there are several different more elaborate approaches, for instance $f=f \cdot \frac{\lam_\star-\lam_0}{\lam-\lam_0}$, or utilizing the knowledge that the approximation-error of $\lam$ is of order $\eps^2$. 
	\item[Line 11:] We do not need the eigenfunctions on the whole domain. In fact, we only need those on the Steklov part of the boundary. Actually, we only need their evaluation on $\mathsf{L}$, but in order to determine the eigenfunctions we also need to norm them and to orthogonalize them in the $\Leu^2(\GSte)$ sense. This is done with the Gram-Schmidt rule together with the trapezoidal rule for discretized eigenfunctions.
	\item[Line 12:] The idea follows from the approximation $\lam_\star \approx \lam + 2\lam \eps\sum_{i=0}^{N-1} u_i(\mathsf{L})^2$. The dampening factor $f$ is discussed in Line $10$.
\end{itemize}

\begin{remark}
	When the function $\del\Om\ni z\mapsto\big(\mathrm{S}^\lam_{\del\Om}(\xS,z) \cdot \mathrm{S}^\lam_{\del\Om}(y,z)\big)\in\RR$ oscillates strongly on the boundary it might yield better results, when multiple, but smaller, boundary intervals are applied. The idea behind this is that using one long boundary interval might intersect the disadvantageous part of the function $\mathrm{S}^\lam_{\del\Om}(\xS,z) \cdot \mathrm{S}^\lam_{\del\Om}(y,z)$ and thus decrease the intensity of $\mathrm{S}^\lam_{\del\Om}(\xS,y)$. This methodology is not investigated in this paper.
\end{remark}
\newcommand{\raystretch}[1]{\renewcommand{\arraystretch}{#1}}

\section{Numerical Implementation and Tests}\label{Ch:NumImplTest}

The numerical implementation for the eigenvalues follows the boundary-integral approach given for the mixed Steklov-Neumann problem in \cite{nurbek,EldarsAlgo} and for the mixed Dirichlet-Neumann problem in \cite{Nigam2014} and is as follows. 

We represent an eigenfunction $u_j\in\mathrm{H}^1(\Om)$ using the single-layer potential, that is,
\begin{align*}
	u_j(x)=\mathcal{S}^0_{\Om}[\phi_j](x)\DEF\int_{\del\Om}\Gamma^0(x\,,y)\phi_j(y)\intd\sigma_y\,,\quad\text{ for } x\in\Om\,,
\end{align*}
where $\Gamma^0(x\,, y)\DEF\frac{1}{2\pi}\log(|x-y|)$ and $\phi_j\in\Leu^2(\del\Om)$ is the potential associated to $u_j$. Using the jump-relations given in \cite[Section 2.2]{LPTSA}, we can obtain that for $x\in\del\Om$, 
\begin{align*}
	u_j(x)
		&=\int_{\del\Om}\Gamma^0(x\,,y)\phi_j(y)\intd\sigma_y\,,\\
	\del_\nu u_j(x)
		&=(-\tfrac{1}{2}\mathcal{I}+(\mathcal{K}^0_{\del\Om})^\ast)[\phi_j](x)\,,\\
		&=-\tfrac{1}{2}\phi_j(x) + \int_{\del\Om}\del_{\nu_x}\Gamma^0(x\,,y)\phi_j(y)\intd\sigma_y.
\end{align*}
We first assume that the domain $\Om$ possesses a $2\pi$-periodic counterclockwise parametric representation of the form $x(t)=(x_1(t)\,, x_2(t))$, for $t\in [0,2\pi]$ and then discretize the interval $[0,2\pi]$, according to the partitioning $(\GSte\,,\GNeu)$ of the boundary. Finally, we use a spectral method based on logarithmic singularity resolution and Fourier series, as discussed in \cite{Nigam2014}. 
To be more specific, we first use an affine linear transformation of the cosine substitution in the aforementioned integrals and then make  use of a Fourier transformation of the resulting integrand. This leads to well-known integrals with analytic expressions. In order to compute the Fourier transformation, a discretization of the underlying interval is used. Back-substitution recovers the initial discretization. This leads us to the discretized operator expressions $\texttt{S}$ and $-\frac{1}{2}\texttt{I} + \texttt{K}^\ast$. Using the underlying boundary conditions given in (\ref{pde:SteklovHom}), we obtain 
\begin{align*}
	(-\tfrac{1}{2}\texttt{I} + \texttt{K}^\ast) \texttt{X} 
		= \text{\footnotesize{$\lambda$}} \, 
		\begin{bmatrix} 
			\texttt{0} \\ \texttt{S}
		\end{bmatrix}
		\texttt{X},
\end{align*}
where $\texttt{X}$ is the discretized expression of the potential $\phi_j$. This equation is solved with the in-build method \verb+eig( , )+ in MATLAB. After obtaining \verb+X+ , we can reconstruct the eigenfunction $u_j$ from the single-layer potential. In order to get the Steklov-Neumann function $\mathrm{S}^\lam$,  we use the decomposition $\mathrm{S}^\lam=\Gamma^0+\mathrm{R}^\lam$, where we determine $\mathrm{R}^\lam$ using
\begin{align*}
	(-\tfrac{1}{2}\texttt{I} + \texttt{K}^\ast) \texttt{R} 
		- \text{\footnotesize{$\lambda$}} \, 
		\begin{bmatrix} 
			\texttt{0} \\ \texttt{S}
		\end{bmatrix}
		\texttt{R}
		=	\begin{bmatrix} 
				-\del\Gamma^0 \\ -(\del\Gamma^0-\lambda\Gamma^0)
			\end{bmatrix}\,,
\end{align*}
where $\texttt{R}$ describes the discretized form of $\mathrm{R}^\lam$, which is analogous to $\texttt{X}$. This is done with the in-build method \verb+mldivide+ in MATLAB, also known as the "\textbackslash "- operator.

Our first numerical test shows the algorithm in the best case scenario. We have the domain $\Omega=\{x\in\RR^2 \mid \NORM{x}_{\RR^2}<1\}$, the source point $\xS \in \{(0,0)^\mathrm{T},(-0.9,0)^\mathrm{T}\}$, the target eigenvalue $\lam_\star=2.5$ and $C_\mathrm{tol}=10^{-3}$. We remark here that the next lower Steklov eigenvalue is a double one at $2$. We let the receiving point $y\in\{ (0,r)^\mathrm{T}\in\RR^2 \mid r > 0\}$ vary. Here we mention that our implementation yields a minuscule imaginary part, because of the MATLAB implementation of \verb+eig+ and the "\textbackslash "- operator, and to avoid amplification of the error during optimization we always projected to the real part. The number of discretization points is $3\cdot256$.
	As it turns out, for $\xS=(0,0)^\mathrm{T}$, the Steklov-Neumann function is independent of $\lam$ and the partitioning because $\Gamma^0$ is zero  on the whole domain. For $\xS=(-0.9,0)^\mathrm{T}$ the Steklov-Neumann function changes and the resulting values are shown in Table \ref{table:1}. An illustration of the Steklov-Neumann function is given in Figure \ref{fig:Circ1}.
	The same set-up but with target eigenvalue $\lam_\star=15.5$ and $3\cdot 512$ discretization points is shown in Table \ref{table:2}.
	
	Our second numerical test shows the algorithm on a non-convex kite-shaped domain $\Om$. This domain is given by the boundary parametrization
\begin{align*}
	\begin{bmatrix}
		\cos(t)+0.65\cos(2 t)-0.65 \\
     	1.5\sin(t)
	\end{bmatrix}\,,
\end{align*} 
	for $t\in [0,2\pi]$. The source point is $\xS = (-1.25,1.25)^\mathrm{T}$, the receiving point is $y = (-1.25,-1.25)^\mathrm{T}$, the target eigenvalue is $\lam_\star=2.5$ and $C_\mathrm{tol}=10^{-3}$. The next lower Steklov eigenvalue is a single one at approximately $2.043996$. The number of discretization points is $6\cdot256$. The resulting Steklov-Neumann function is shown in Figure \ref{fig:Kite1}.
	The center of the Neumann boundary condition $\GNeu$ is at $\sim (0.257,-0.947)^\mathrm{T}$ with length $\sim 1.489$. $\mathrm{S}_{\text{Steklov}}^\lam(\xS,y)\approx 0.0199$ and $\mathrm{S}_{\text{End}}^\lam(\xS,y)\approx -124.1$.

\begin{figure}[h]
  \begin{subfigure}{0.49\textwidth}
    \centering
    \includegraphics[width=0.99\textwidth]{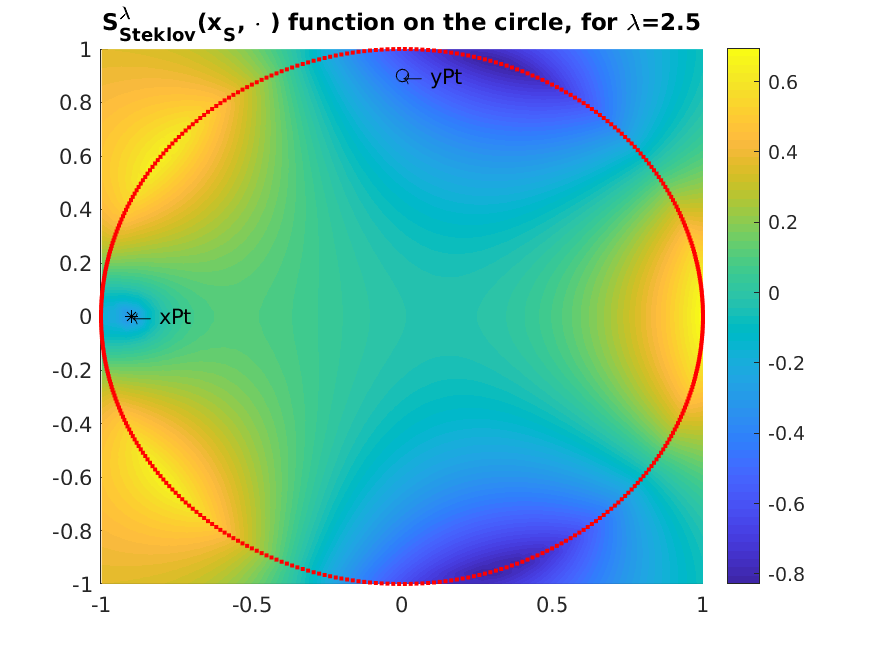}
  \end{subfigure}\hfill 
  \begin{subfigure}{0.49\textwidth} 
    \centering
    \includegraphics[width=0.99\textwidth]{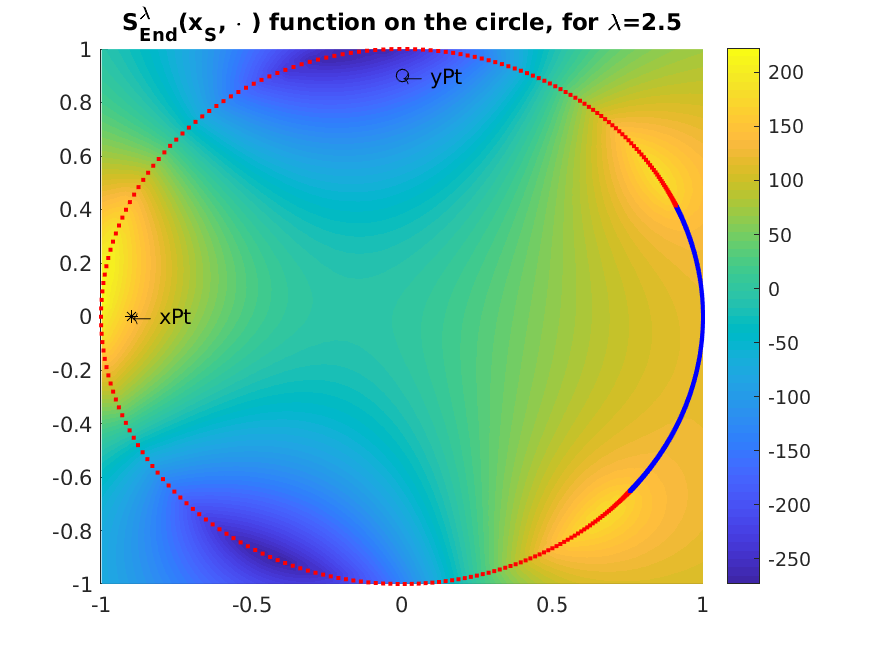}
  \end{subfigure}
  \caption{The Steklov-Neumann function for $\lam_\star=2.5$ on the unit disk with Steklov boundary condition on the left and final mixed boundary conditions on the right. Marked are $\xS$, denoted as 'xPt', and $y$, denoted as 'yPt'. The points on the boundary are our discretization points. Blue points denote the Neumann boundary conditions, red points denote the Steklov boundary conditions.}\label{fig:Circ1}
\end{figure}

\begin{table*}\centering
\raystretch{1.1}
\resizebox{\columnwidth}{!}{%
\begin{tabular}{@{}rrrrrr@{}}
	& $r=0.1$ & $r=0.25$ & $r=0.5$ & $r=0.75$ & $r=0.9$\\ 
	\midrule
$\mathrm{S}_{\text{Steklov}}^{\lam_\star}(\xS,y)$ 
	& -0.022
	& -0.048
	& -0.147 
	& -0.332 
	& -0.492 \\
$\mathrm{S}_{\text{End}}^{\lam_\star}(\xS,y)$ 
	& 12.67
	& 39.83
	& 90.50
	& -133.6
	& -200.4 \\
$\left|\frac{\mathrm{S}_{\text{End}}^{\lam_\star}(\xS,y)}{\mathrm{S}_{\text{Steklov}}^{\lam_\star}(\xS,y)}\right|$ 
	& 586
	& 838
	& 615
	& 402
	& 407 \\
$\theta_{\mathrm{center}}$ 
	& $0.42 \pi$
	& $0.42 \pi$
	& $0.39 \pi$
	& $1.96 \pi$
	& $1.96 \pi$ \\
$l_{\mathrm{N}}$ 
	& $0.36 \pi$
	& $0.36 \pi$
	& $0.36 \pi$
	& $0.36 \pi$
	& $0.36 \pi$ \\
\bottomrule
\end{tabular}
}
\caption{Algorithm \ref{Algorithm} tested on the unit circle with $\lam_\star = 2.5$, $\xS = (-0.9,0)^\mathrm{T}$, $y\in\{ (0,r)^\mathrm{T}\in\RR^2 \mid r > 0\}$, $C_\mathrm{tol}=10^{-3}$. $\mathrm{S}_{\mathrm{Steklov}}^\lam(\xS,y)$ represents the Steklov function and $\mathrm{S}_{\mathrm{End}}^\lam(\xS,y)$ represents the Steklov-Neumann function on the final partition, where the final partition is made out of two boundary intervals, one with Steklov boundary conditions and the other with Neumann boundary conditions. $\theta_{\mathrm{center}}\in [0, 2 \pi)$ represents the angle of the center of the Neumann boundary intervals and $l_{\mathrm{N}}$ its length.}\label{table:1}
\end{table*}

Interpreting the test data, we infer that adding Neumann boundary intervals moves the peaks at the boundary and can increase the overall value of the Steklov-Neumann function as seen in Figure \ref{fig:Circ1}. Consider that on the Neumann boundary interval itself, there are no significant peaks and the peaks which were on that part of the boundary before, get pressed into the remaining part of the boundary. Hence, it seems that the algorithm does not only close in to the target eigenvalue $\lam_\star$, but also moves one of the peaks close to the receiving point. This is feasible as long as the receiving point is close enough to the boundary, since no peaks can be found in the interior of the domain.

\begin{table*}\centering
\raystretch{1.1}
\resizebox{\columnwidth}{!}{%
\begin{tabular}{@{}rrrrrr@{}}
	& $r=0.1$ & $r=0.25$ & $r=0.5$ & $r=0.75$ & $r=0.9$\\ 
	\midrule
$\mathrm{S}_{\text{Steklov}}^{\lam_\star}(\xS,y)$ 
	& -0.166
	& -0.016
	& -0.014 
	& -0.013 
	& -0.038 \\
$\mathrm{S}_{\text{End}}^{\lam_\star}(\xS,y)$ 
	& -0.540
	& -1.611
	& -5.811
	& -9.719
	& -38.04 \\
$\left|\frac{\mathrm{S}_{\text{End}}^{\lam_\star}(\xS,y)}{\mathrm{S}_{\text{Steklov}}^{\lam_\star}(\xS,y)}\right|$ 
	& 32.5
	& 101
	& 418
	& 731
	& 1014 \\
$\theta_{\mathrm{center}}$ 
	& $0.52 \pi$
	& $0.52 \pi$
	& $0.52 \pi$
	& $0.52 \pi$
	& $0.98 \pi$ \\
$l_{\mathrm{N}}$ 
	& $0.05 \pi$
	& $0.05 \pi$
	& $0.05 \pi$
	& $0.05 \pi$
	& $0.05 \pi$ \\
\bottomrule
\end{tabular}
}
\caption{Algorithm \ref{Algorithm} tested on the unit circle with $\lam_\star = 15.5$, $\xS = (-0.9,0)^\mathrm{T}$, $y\in\{ (0,r)^\mathrm{T}\in\RR^2 \mid r > 0\}$ and $C_\mathrm{tol}=10^{-3}$. $\mathrm{S}_{\mathrm{Steklov}}^\lam(\xS,y)$, $\mathrm{S}_{\mathrm{End}}^\lam(\xS,y)$, $\theta_{\mathrm{center}}\in [0, 2 \pi)$ and $l_{\mathrm{N}}$ are defined as in Table \ref{table:1}.}\label{table:2}
\end{table*}

\begin{figure}[h]
  \begin{subfigure}{0.49\textwidth}
    \centering
    \includegraphics[width=0.99\textwidth]{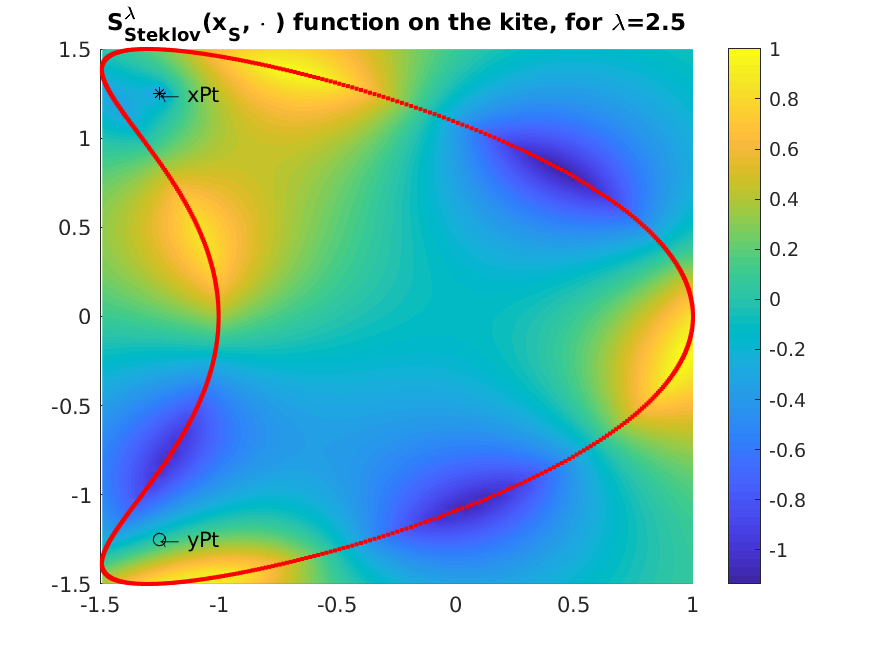}
  \end{subfigure}\hfill 
  \begin{subfigure}{0.49\textwidth} 
    \centering
    \includegraphics[width=0.99\textwidth]{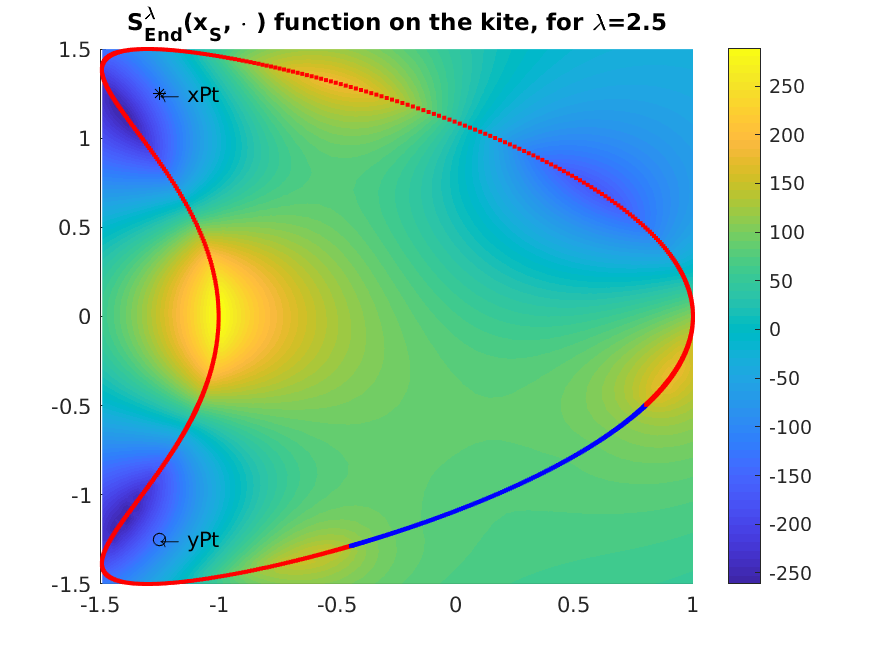}
  \end{subfigure}
  \caption{The Steklov-Neumann function for $\lam_\star=2.5$ on the kite shape with Steklov boundary condition on the left and final mixed boundary conditions on the right. Further notation is the same as in Figure \ref{fig:Circ1}. }\label{fig:Kite1}
\end{figure}

\section{Concluding Remarks}
We have presented asymptotic expressions for the change in the Steklov-Neumann eigenpairs when a small portion of the boundary is changed from Steklov to Neumann. These ideas are exploited to derive a method for optimizing the Green's function for a mixed Steklov-Neumann problem at a given source-receiver pair, and for a value of the Steklov parameter which is close to a given target value. The optimization requires the highly accurate evaluation of such mixed eigenvalues, and this step is performed using a boundary integral approach. The optimization strategy is effective, and can lead to gains of two orders of magnitude.



\bibliographystyle{plain}
\bibliography{refs}

\end{document}